\documentclass[
 reprint,
 amsmath,amssymb,
 aps,longbibliography
]{revtex4-2}

\usepackage{graphicx}
\usepackage{dcolumn}
\usepackage{bm}
\usepackage{braket}

\usepackage[utf8]{inputenc}
\usepackage{tikz}
\usetikzlibrary{calc}
\usetikzlibrary{perspective}
\usetikzlibrary{decorations.markings}

\usepackage{amsmath}
\usepackage{amsthm}
\newtheorem{theorem}{Theorem}

\newtheorem{corollary}{Corollary}
\newtheorem{lemma}{Lemma}

\usepackage[english]{babel}

\usepackage{xr-hyper}

\usepackage[hidelinks, breaklinks]{hyperref}
\usepackage{breakurl}

\usepackage[toc,page]{appendix}

\begin{document}

\preprint{APS/123-QED}

\title{Tip of the Quantum Entropy Cone}

\author{Matthias Christandl}
\email{christandl@math.ku.dk}
\author{Bergfinnur Durhuus}
\email{durhuus@math.ku.dk}
\author{Lasse Harboe Wolff}
 \email{lhw@math.ku.dk}
\affiliation{Department of Mathematical Sciences, University of Copenhagen, Universitetsparken 5, 2100 Denmark}

\date{\today}

\begin{abstract}
Relations among von Neumann entropies of different parts of an $N$-partite quantum system have direct impact on our understanding of diverse situations ranging from spin systems to quantum coding theory and black holes. Best formulated in terms of the set $\Sigma^*_N$ of possible vectors comprising the entropies of the whole and its parts, the famous strong subaddivity inequality constrains its closure $\overline\Sigma^*_N$, which is a convex cone. Further homogeneous constrained inequalities are also known. 

In this work we provide (non-homogeneous) inequalities that constrain $\Sigma_N^*$ near the apex (the vector of zero entropies) of $\overline\Sigma^*_N$, in particular showing that $\Sigma_N^*$ is not a cone for $N\geq 3$. Our inequalities apply to vectors with certain entropy constraints saturated and, in particular, they show that while it is always possible to up-scale an entropy vector to arbitrary integer multiples it is not always possible to down-scale it to arbitrarily small size, thus answering a question posed by A. Winter.
\end{abstract}

\maketitle

\section{Introduction}\label{sec:introduction}

Entropy is a very important concept in physics, whose role and status have vastly expanded past its original boundaries within thermodynamics. It is a main object of study in many areas of research, including quantum cryptography, information theory, black holes and more.    

In models of the world, it is often very advantageous and natural to consider large systems as composed of smaller distinct subsystems. This calls for a good understanding of the relations among entropies of different subsystems of a joint system. The most important such relation is without a doubt the \emph{strong subaddivity} inequality \cite{lieb_ruskai_1973}, which entails all other known entropy inequalities for multi-partite quantum systems and has long been appreciated in quantum information theory. There has naturally been a great interest in finding new such inequalities. The problem of finding new entropy inequalities is an aspect of a more general research endeavour to adequately describe the set of possible values that the different allocations of entropy in a multi-partite system can take, i.e. to determine whether or not any given ordered set of numbers corresponds to an achievable \emph{entropy-vector}, by which we mean the entropy-values of the marginals of some quantum state.  

In this Letter, we prove a new relationship between the entropies of a multi-partite system, which rules out the possibility of constructing certain small entropy-vectors that otherwise satisfy strong subadditivity and related inequalities. This result, interestingly, entails that the set of achievable entropy vectors is neither a cone nor a closed set -- thus answering a question left open in an influential paper by Pippenger \cite{pippenger_inequalities_2003}. We additionally discuss applications of the new results to a diverse set of areas -- namely topological materials, entanglement theory, and quantum cryptography. 

In the remainder of this introduction, we shall introduce some necessary notation and relevant background concerning the quantum entropy cone. The main results are presented in the following section, after which we discuss some applications and provide a conclusion and outlook for this work. 
\\
\\
Given a quantum system $X$ in a state described by a density operator $\rho$, i.e. a non-negative operator of trace $1$ on a (finite dimensional) Hilbert space $\mathcal H_X$, its von Neumann entropy is given by 
\begin{align} \label{eq:def of entropy}
H_\rho = -{\rm Tr}[\rho \log(\rho)]=-\sum_i \lambda_i \log(\lambda_i)\,,
\end{align}
where $\lambda_i$ are the eigenvalues of $\rho$, and $\log$ denotes the binary logarithm. We shall be concerned with multipartite systems $\mathcal N$ consisting of   
 $N$ constituent systems $X_1,...,X_N$ with associated Hilbert spaces $\mathcal{H}_{X_1},...,\mathcal{H}_{X_N}$, such that the state of $\mathcal N$ is given by a density operator $\rho$ on $\mathcal{H}_{X_1} \otimes ... \otimes \mathcal{H}_{X_N}$. 
The reduced state of a subsystem $\mathcal{X} \subseteq \mathcal{N}$ is then given by
\begin{align*}
\rho_{\mathcal{X}}:={\rm Tr}_{\mathcal{N} \backslash \mathcal{X}}[\rho]\,,
\end{align*}
where ${\rm Tr}_{\mathcal{N} \backslash \mathcal{X}}[ \cdot ]$ denotes the partial trace over $\otimes_{X_i\notin\mathcal{X}} \mathcal{H}_{X_i}$ (and in particular $\rho=\rho_{\mathcal N}$). The entropy $H_{\rho_{\mathcal X}}$ of the reduced state  will also be denoted by $H(\mathcal X)_{\rho}$ or by $H(X_{i_1}\dots X_{i_k})_{\rho}$, if $\mathcal X=\{X_{i_1},\dots, X_{i_k}\}$. 
These marginal entropies define a vector $\vec{H}_{\rho} \in \mathbb R^{2^N-1}$, called the \emph{entropy vector} of $\rho$, whose coordinates are labelled by the non-empty subsystems of $\mathcal N$. E.g., for $N=2$ and $\mathcal N=\{A,B\}$ we have $\vec{ H}_{\rho}= (H(A),H(B), H(AB))_{\rho} \in \mathbb R^3$, while for $N=3$ and $\mathcal N=\{A,B,C\}$ we write 
\begin{align}\label{vector}
\vec{H}_{\rho} \; = & \;(H(A), H(B), H(C),H(BC),\nonumber\\  & \quad \ H(AC),H(AB), H(ABC))_{\rho}\;\in \mathbb R^7\,.
\end{align} 
The main object of study in this context is the set  $\Sigma_N^*$ of all possible entropy vectors associated to $N$-partite systems,
$$
\Sigma_N^* = \{\vec{H}_{\rho} \in \mathbb{R}^{2^N-1} \mid \rho\; \mbox{is a density operator on $\mathcal N$}\}\,.
$$
It is a fundamental result of Pippenger \cite{pippenger_inequalities_2003} that the topological closure $\overline\Sigma_N^*$ of $\Sigma_N^*$ in $\mathbb R^{2^N-1}$ is a convex cone, called the \emph{quantum entropy cone} of $N$-partite systems, i.e. $\overline\Sigma_N^*$ is closed under addition and under multiplication by positive scalars. It is also known, and easy to demonstrate, that $\Sigma_N^*$ has full dimension, i.e. it spans all of $\mathbb R^{2^N-1}$ as a vector space, and that $\overline\Sigma_N^*$ and $\Sigma_N^*$ have identical interiors and hence also identical boundaries. For $N=2$ it is even true that $\overline\Sigma_2^*=\Sigma_2^*$ as will be commented on further below. But for general $N\geq 3$ an appropriate characterisation of the boundary entropy vectors is missing \footnote{A restricted class of boundary states has been studied in \cite{majenz2018constraints}}.

A related but different long standing problem is to determine whether or not $\overline\Sigma_N^*$ is a polyhedral cone, i.e. if it can be specified in terms of a finite number of linear inequalities. The known general inequalities of this sort are of two types: 
\begin{align}  
& \qquad H(\mathcal X)+H(\mathcal Y) \geq H(\mathcal X\cap \mathcal Y)+H(\mathcal X\cup \mathcal Y) \label{eq: def ssa}\\
& \qquad H(\mathcal X)+H(\mathcal Y)\geq H(\mathcal X\setminus\mathcal Y)+H(\mathcal Y\setminus\mathcal X)\,,\label{eq: def weak monotonicity}
\end{align}
called \emph{strong subadditivity} and \emph{weak monotonocity}, respectively. Here, $\mathcal X$ and $\mathcal Y$ are arbitrary subsystems, and by convention we have $H(\emptyset)=0$. We emphasize that not all inequalities of the forms above are independent. Strong subadditivity was first established in \cite{lieb_proof_1973}, but a variety of proofs exist in the literature, see e.g. \cite{lieb_ruskai_1973, ruskai_another_2007-1, nielsen_simple_2005, nielsen_quantum_2010, watrous_theory_2018, christandl_recoupling_2018}. To obtain weak monotonicity one makes use of the fact, referred to as \emph{purification} \cite{nielsen_quantum_2010}, that given a state $\rho$ of $\mathcal{N}$ it is always possible to extend $\mathcal N$ by a system $Y$ and to define a pure state $\eta =\ket{V} \bra{V}$ of $\mathcal N\cup Y$ such that $\rho = \eta_{\mathcal N}$.

The polyhedral cone defined by \eqref{eq: def ssa} and \eqref{eq: def weak monotonicity} is a closed convex cone, and will here be denoted $\Sigma_N$. The question of whether $\Sigma_N=\overline\Sigma_N^*$, or if there exist further independent linear inequalities beyond  \eqref{eq: def ssa} and \eqref{eq: def weak monotonicity}, remains open for $N\geq 4$. For $N\leq 3$ the two closed cones coincide as shown in \cite{pippenger_inequalities_2003}. While it is quite easy to see that $\Sigma_N=\Sigma_N^* = \overline\Sigma_N^*$ hold for $N\leq 2$, the case $N\geq 3$ is different. It has been shown that for $N \geq 4$ there exist further constrained homogeneous linear inequalities \cite{linden_new_2005, cadney_infinitely_2012, majenz_entineq_2018}. 
\\
\\
We shall now delve a bit deeper into the details of the case $N=3$ where the relevant inequalities are
\begin{align}
&I_{XY} := H(X)+H(Y)-H(XY)\geq 0\notag\\
&II_{XY} := H(XZ)+H(YZ)-H(Z) - H(XYZ) \geq 0 \notag\\
&III_{XY} := H(Z)+H(XYZ)- H(XY)\geq 0 \notag\\
&IV_{XY} := H(XZ) + H(YZ) - H(X) - H(Y)\geq 0 \notag\,,
\end{align}
valid for $\{X,Y\}$ equaling $\{A,B\}, \{A,C\}$ or $\{B,C\}$ with $Z\neq X,Y$. This makes a total of twelve inequalities, three of each type. A key observation is that
\begin{equation}\label{eq: M}
M := I_{XY} - II_{XY} = III_{XY} - IV_{XY}
\end{equation}
is independent of the choice of $\{X,Y\}$. It follows that $\Sigma_3$ is a union of two cones 
\begin{align}
 \Sigma_3^+ :\quad &II_{XY}\geq 0\,,\quad IV_{XY}\geq 0\,,\quad M\geq 0\\
\Sigma_3^- :\quad &I_{XY}\geq 0\,,\quad III_{XY}\geq 0\,,\quad M\leq 0\,,   
\end{align}
each of which has seven facets, corresponding to their seven defining inequalities. 

By a slight elaboration of Pippenger's approach \cite{pippenger_inequalities_2003} it can be shown that $\Sigma_3^+\subset\Sigma_3^*$, while $\Sigma_3^-$ behaves differently. For any $\vec H\in \Sigma_3^-$ one finds that there exists a quantum state $\rho$ and a vector $\vec{l}$ belonging to the $1$-dimensional face (half-line) $\ell$ of $\Sigma_3^-$ defined by the six equations
\begin{equation}\label{eq: line1}
\ell \ : \qquad  I_{XY} = 0\,,\qquad III_{XY} = 0 \, ,  
\end{equation}
such that
\begin{align}\label{eq: Hell}
\vec H = \vec{l} + \vec{H}_{\rho}\,.    
\end{align}
If it so happened that $\ell\subset\Sigma_3^*$, it would follow by the additivity of entropy vectors in suitably constructed product states that $\Sigma_3^-\subset\Sigma_3^*$ and hence that $\Sigma_3=\Sigma_3^*$. However, as a consequence of Theorem \ref{thm:most general constrained inequality} below there is an open line segment of $\ell$ ending at the apex which is not contained in $\Sigma_3^*$, and so $\Sigma_3\neq \Sigma_3^*$. 
On the other hand, Pippenger identifies a state $\rho^l$ such that $\vec{H}_{\rho^l} \in \ell$, which by the cone property implies that $\ell\subset\overline\Sigma^*_3$. Using \eqref{eq: Hell} one then obtains that $\Sigma_3^-\subset \overline\Sigma_3^*$ and consequently $\Sigma_3 =\overline\Sigma_3^*$, which is the already mentioned  main result of \cite{pippenger_inequalities_2003}. 

In order to satisfy \eqref{eq: line1}, the entropy vector  $\vec{H}_{\rho^l}$ must satisfy
\begin{equation} 
I(X:Y)_{\rho^l}=0\,,\ 
H(X)_{\rho^l}+H(XYZ)_{\rho^l}= H(YZ)_{\rho^l}\label{eq: line2}
\end{equation}
for any pair $\{X,Y\}$ in $\mathcal N=\{A,B,C\}$ with $Z\neq X,Y$, where the more standard notation $I(X:Y)$ has been used instead of $I_{XY}$ for the \emph{quantum mutual information}. By purification one can alternatively consider a state $\eta=\ket{V} \bra{V}$ on a $4$-partite system $\{A,B,C,D\}$ such that $\rho^l=\eta_{\mathcal N} $. Such a pure state makes
the equations \eqref{eq: line2} take on the more symmetric form 
\begin{equation}\label{eq: line3}
I(X_i:X_j)_{\eta}=0
\end{equation}
for all pairs $X_i,X_j$ in $\{A,B,C,D\}$. Indeed, the state $\rho ^l$ is obtained in \cite{pippenger_inequalities_2003} by first constructing such a pure state $\eta$. Our main theorem below concerns pure states of arbitrary $N$-partite systems that fulfill the conditions (\ref{eq: line3}) for fixed $i$, showing that sufficiently small scalar multiples of their entropy vectors lie outside $\Sigma_N^*$ i.e. cannot be realized by quantum states. For the sake of completeness we exhibit in Appendix \ref{app:explicit expressions for states}, for arbitrary $N\geq 4$, states which fulfill the stated conditions, and thus generalising the pure state $\eta $ mentioned above.

\section{Main results} \label{sec:presenting entropy inequality}

The goal of this section is to establish the following entropy bound.
\begin{theorem} \label{thm:most general constrained inequality}
Let $\rho$ be a pure state of the $N$-partite system $\mathcal{N}=\{X_1,...,X_N\}$ such that $H(X_1)_{\rho} \neq 0$. Suppose further that 
\begin{align*}
I(X_1:X_i)_{\rho}=0 \quad \text{for all} \quad i= 2, \dots, N\,.
\end{align*}
Then the following bound holds:
\begin{align} \label{eq:1st ineq central thm}
\sum_{i=1}^N & H(X_i)_{\rho} \ > \ 1\,.
\end{align}
\end{theorem}
The conditions in the theorem are illustrated in Fig. \ref{figure: representation of the constraints}. Note that they can only be satisfied if $N \geq 4$.

To establish Theorem \ref{thm:most general constrained inequality}, we first list three lemmas below which are the main ingredients in the subsequent proof. Their demonstrations are provided in Appendix \ref{app:proof of lemmas}.
We will use the following notation. Given a state $\rho$ of $\mathcal{N}$, we denote by $\lambda_1^i\geq\lambda_2^i \geq \dots$ the eigenvalues of $\rho_{X_i}$ in decreasing order and by $\ket{e^i_1},\ket{e^i_2},\dots$ a corresponding orthonormal eigenstate basis such that 
\begin{align} \label{eq:diagonalized density matrix}
(\rho_{X_i})_{a b}:=\bra{e^i_a} \rho_{X_i} \ket{e^i_b} = \lambda^i_a \delta_{ab}\,.
\end{align}
 Moreover, we define 
\begin{equation}\label{eq: epsilons}
\epsilon_i:=1-\lambda_1^i\quad\mbox{and} \quad \varepsilon:=\sum_{i=1}^N \epsilon_i\,.
\end{equation}
Clearly, $\sum_{x_i > 1} \lambda_{x_i}^i=\epsilon_i$ and one easily verifies that
\begin{equation} \label{eq:entropy bound larg eigenvalue}
H(X_i) \geq \max \{ h(\epsilon_i), \ - \log(1-\epsilon_i) \} \geq 2 \epsilon_i\,,
\end{equation}

\begin{figure}

\begin{tikzpicture}[scale=2,double distance=3pt,decoration={markings,
  mark=between positions 0 and 1 step 6pt
  with { \draw [fill] (0,0) circle [radius=0.35pt];}}]
\coordinate (Rho1) at (0,0);
\coordinate (Rho3) at (-0.707,0.707);
\coordinate (RhoN) at (1,0);
\coordinate (Rho2) at (-1,0);
\coordinate (Rho1N) at (0.707,0.707);

\path[postaction={decorate}] (Rho3) to (Rho1N);

\draw[double,thick] (Rho1) -- (RhoN);
\draw[double,thick] (Rho1) -- (Rho1N);
\draw[double,thick] (Rho1) -- (Rho3);
\draw[double,thick] (Rho1) -- (Rho2);

\filldraw[white] (Rho1) circle (10pt);
\filldraw[white] (Rho2) circle (10pt);
\filldraw[white] (Rho3) circle (10pt);
\filldraw[white] (RhoN) circle (10pt);
\filldraw[white] (Rho1N) circle (10pt);

\filldraw[.style=draw,fill=white] (Rho1) circle (6.5pt);
\filldraw[.style=draw,fill=white] (Rho2) circle (6.5pt);
\filldraw[.style=draw,fill=white] (Rho3) circle (6.5pt);
\filldraw[.style=draw,fill=white] (RhoN) circle (6.5pt);
\filldraw[.style=draw,fill=white] (Rho1N) circle (6.5pt);

\node[scale=1] at (Rho1) {$X_1$};
\node[scale=1] at (Rho3) {$X_3$};
\node[scale=1] at (RhoN) {$X_N$};
\node[scale=1] at (Rho2) {$X_2$};
\node[scale=1] at (Rho1N) {$X_{N-1}$};

\end{tikzpicture}

\caption{The conditions of Theorem \ref{thm:most general constrained inequality} are here represented with each circle denoting a constituent system $X_i$. The double lines indicate that the mutual information between the two systems is $0$, and it is assumed that the total state is pure.} \label{figure: representation of the constraints}
\end{figure}
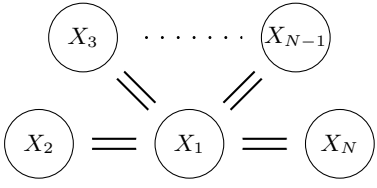
where $h$ denotes the \emph{binary entropy function},
$$
h(x) = -x\log x-(1-x)\log(1-x)\,. 
$$
Assuming  $\rho$ to be pure, i.e. $\rho=\ket{V} \bra{V}$ where $\langle V| V\rangle =1$, we represent $|V\rangle$ with respect to the basis for $\mathcal H_{\mathcal N}$ consisting of tensor products of eigenstates $|e^i_a\rangle$ for the single-party density matrices, that is 
\begin{align} \label{eq: components of V}
\ket{V} = & \sum_{x_1,\dots,x_N}V_{x_1 ... x_N}|e^1_{x_1}\dots e^N_{x_N}\rangle\,, \\
\text{where} \quad & \sum_{x_1,\dots,x_N} |V_{x_1 ... x_N}|^2\, = \,1\,. \label{eq: V has norm 1}
\end{align}
A sum over dummy-indices $x_i \in \mathbb{N}$ will here always run up to $\mathrm{dim}(\mathcal{H}_{X_i})$. The matrix elements of $\rho_{\mathcal N}$ and the reduced states are quadratic expressions of the components of $\ket{V}$; e.g., 
\begin{align} \label{eq: density matrix 1 in terms of V components}
(\rho_{X_1})_{a \ b} & = \sum_{x_2,...,x_N} V_{a x_2 ... x_N} V_{b x_2 ... x_N}^* \\  \label{eq: bipartite density matrix 1 2 in terms of V components}
(\rho_{X_1 X_2})_{a_1 a_2 \ b_1 b_2}  
& = \sum_{x_3,...,x_N} V_{a_1 a_2 x_3... x_N} V_{b_1 b_2 x_3... x_N}^*\,.
\end{align}
Extensive use will be made of the fact that $I(X_i:X_j)=0$ holds if and only if $\rho_{X_iX_j}$ is a product state, which in our notation and choice of basis means that
\begin{align} \label{eq:components of relevant product state for main thm}
(\rho_{X_iX_j})_{a_1 a_2 \ b_1 b_2} =\lambda_{a_1}^i \lambda_{a_2}^j \delta_{a_1 b_1} \delta_{a_2 b_2}\,.
\end{align}
The announced lemmas relate the $\epsilon_i$'s to the components of $V$ as follows.

\begin{lemma} \label{lemma: lower bound V1111}
For any pure state $\rho$ it holds that
\begin{align} \label{eq:lower bound V1111}
|V_{1...1}|^2 \geq 1-\varepsilon.
\end{align}
\end{lemma}

\begin{lemma} \label{lemma: lower bound Va111}
For any pure state $\rho$ such that $I(X_1:X_j)=0$ for all $j \neq 1$ we have
\begin{align} \label{eq:lower boung Vi111}
\sum_{x_1>1} |V_{x_1 1 ... 1}|^2 \geq \epsilon_1 (1+\epsilon_1-\varepsilon)\,.
\end{align}
\end{lemma}

\begin{lemma} \label{lemma:product inequality}
For any pure state $\rho$ it holds that
\begin{align} \label{eq:product inequality eq}
(1-\epsilon_1) \sum_{x_1>1} |V_{x_1 1 ... 1}|^2 \leq \epsilon_1 (\varepsilon-\epsilon_1)\,.
\end{align}
\end{lemma}

We remark that Lemma \ref{lemma: lower bound V1111} is used for the proof of Lemma \ref{lemma:product inequality}, while only Lemma \ref{lemma: lower bound Va111} and Lemma \ref{lemma:product inequality} are used in the proof of Theorem \ref{thm:most general constrained inequality}.

\begin{proof} [Proof of Theorem \ref{thm:most general constrained inequality}.] \;
Combining Lemma \ref{lemma: lower bound Va111} and Lemma \ref{lemma:product inequality} we get
\begin{align*}
(1-\epsilon_1) \epsilon_1 (1+\epsilon_1-\varepsilon) \, \leq \, \epsilon_1 (\varepsilon-\epsilon_1)\,.
\end{align*}
Since $\epsilon_1 > 0$ as a consequence of the assumption $H(X_1) \neq 0$, this is equivalent to  
\begin{align*}
1 + (1+\varepsilon-\epsilon_1) \epsilon_1 \leq 2\varepsilon\,.
\end{align*}
Since the left-hand side of this inequality is larger than $1$, it follows that
$\varepsilon > \frac 12$ which in turn implies   \eqref{eq:1st ineq central thm} by use of \eqref{eq:entropy bound larg eigenvalue} and the definition of $\varepsilon$. This completes the proof of Theorem \ref{thm:most general constrained inequality}.

\end{proof}

\def\newscale{7}
\def\endlength{2.2}

\begin{figure}

\begin{tikzpicture}[3d view={-10}{20},scale=\newscale]

\coordinate (O) at (0,0,0);

\coordinate (C1) at ($0.25*(1,1,1)$);
\coordinate (C2) at ($0.5*(1,0,0)$);
\coordinate (C3) at ($0.25*(1,2,0)$);

\coordinate (E1) at ($0.25*\endlength*(1,1,1)$);
\coordinate (E2) at ($0.5*\endlength*(1,0,0)$);
\coordinate (E3) at ($0.25*\endlength*(1,2,0)$);

\filldraw[draw=none][fill=cyan!30!white,
]
(C1) -- (E1)
{ -- (E2)}
-- (C2);

\filldraw[draw=none][fill=cyan!40!white,
]
(C2) -- (E2)
{ -- (E3)}
-- (C3);

\filldraw[draw=none][fill=cyan!40!white,
]
(C3) -- (E3)
{ -- (E1)}
-- (C1);

\draw[draw=none][fill=cyan!54!white,
] (C1) -- (C2) -- (C3);

\draw[cyan!61!white,thick] (C3) -- ($1.04*(E3)$);

\draw[dashed][red!70!white,ultra thick] (O) -- (C1);
\draw[dashed][thick] (O) -- (C2);
\draw[dashed][thick] (O) -- (C3);
\draw[thick] (C1) -- (C2) -- (C3) -- cycle;

\draw[red!70!white,ultra thick] (C1) -- ($1.04*(E1)$);
\draw[thick] (C2) -- ($1.03*(E2)$);

\filldraw[black] (O) circle (0.25pt);
\node[scale=1.8][yshift=0.4cm,xshift=-0.15cm] at (O) {O};

\node[red][scale=1.8][yshift=0.25cm,xshift=-0.1cm] at ($1.8*(C1)$) {$\ell$};

\end{tikzpicture} 
\caption{The solid figure represents the set of permissible values for $( H(A),H(C),H(ABC) )$ satisfying $III_{XY}=0$ for all $X,Y$, given the the inequalities (\ref{eq: def ssa}) and (\ref{eq: def weak monotonicity}) and Corollary \ref{cor: constrained inequality N constarints}.
We have further made the projection $H(A)=H(B)$ to get a $3$-dimensional surface. The dashed lines span a part of $\Sigma_3$ ruled out by Corollary \ref{cor: constrained inequality N constarints}, and $O$ denotes the apex of $\Sigma_3$. The ray $\ell$ is the top edge in the figure.} \label{figure: entropy cone with line l}
\end{figure}
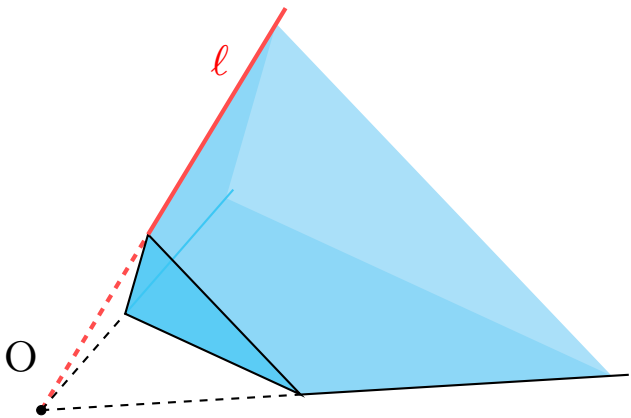

In case the given state $\rho$ is not pure, we can apply Theorem 1 to its purification and obtain (see Appendix \ref{app:lower bound for max entropy}
for more details and further elaboration of the main theorem)

\begin{corollary} \label{cor: constrained inequality N constarints}
Let $\vec{H}$ be a realizable entropy-vector for a system   $\mathcal{N}=\{X_1,...,X_N\}$ which fulfills
$$H(\mathcal{N})>0\;\; \mbox{and}\;\; H(X_i)+H(\mathcal{N})=H(\mathcal{N} \backslash X_i)$$
 for all $i \in \{1,...,N\}$. Then the following bound holds:
\begin{align}\label{newconstraint}
H(\mathcal{N})+\sum_{i=1}^N  H(X_i) & \ > \ 1\,.
\end{align}
\end{corollary}
We note that the conditions in the corollary can be satisfied if $N \geq 3$. This result excludes a range of vectors in $\Sigma_N$ from $\Sigma_N^*$ that satisfy $N$ linear constraints and hence can be labeled by $2^N-N-1$ parameters. See Figure \ref{figure: entropy cone with line l} for a visualization in case $N=3$. In Appendix \ref{app:explicit expressions for states} we provide a $4$-parameter family of realizable entropy vectors on the boundary of $\Sigma_N$ satisfying the conditions of the corollary.

\section{Applications} \label{sec:application to physics}

The entropy concept itself originally arose from thermodynamical considerations of macroscopic systems consisting of many particles, such as gases. Quantum correlations of such systems can be quantified in terms of the scaling of the \emph{entanglement entropy}, that is the entropy of a subregion $A$. It has been found for many systems that this entropy is roughly proportional to the size of the boundary $\partial A$ and not to the volume, a statement known as the \emph{area law} \cite{eisert-area-law-review}. For topologically ordered systems it is expected that
$$H(A)=\alpha\, |\partial A| -\gamma$$
up to terms vanishing as the "area" $|\partial A|$ gets large. Moreover, the constant additive term $-\gamma$ is expected to be universal and is dubbed the \emph{topological entanglement entropy}. Actually, $-\gamma$ equals an alternating sum of entropies, called $M$, encountered above in \eqref{eq: M}. As shown in  \cite{kitaev_topological_2006, levin_detecting_2006} the value of $\gamma$ in a class of systems is always positive, and $M$ is thus negative. This is precisely the regime in which we identified restrictions on entropy vectors and they may therefore have implications for the attainable values of the topological entropy. We point out, however, that the entropy vectors of the particular finite systems calculated in \cite{kitaev_topological_2006, levin_detecting_2006} in terms of their total quantum dimension do not satisfy the conditions of our theorem. Also, as the constraints we obtained are not balanced \cite{cadney_infinitely_2012}, our results have no direct bearing on the usual situation when a large system size is considered.

Many functions in quantum information theory are defined in terms of optimizations of von Neumann entropies \cite{levin-smith-optimized} or even optimization with entropic constraints \cite{bottleneck}. An example from entanglement theory is the squashed entanglement \cite{squashed}
$$E_{sq}(\rho_{AB})=\inf \frac{1}{2} \left( I(A:BE)_{\rho}-I(A:B)_{\rho} \right)\,,$$
where the minimization is over extensions $\rho_{ABE}$ of $\rho_{AB}$. The results of the present work constrain such optimization 
and it remains to be explored whether they could lead to simplified computations in specific cases. 

Finally, let us consider a cryptographic situation, known as quantum secret sharing \cite{hillary, cleve, gottesman}: Alice (A) wishes to distribute information to $N-1$ parties ($N\geq 4$) 
\begin{itemize}
    \item purely, in the sense that the overall state of her and the constituent systems is pure
    \item secretly, in the sense that every share is in product with hers
    \item non-trivially, in the sense that $H(A)>0\,.$
\end{itemize}
These are precisely the conditions of Theorem \ref{thm:most general constrained inequality} and thus it follows from our work that she cannot do so unless the average share carries a minimum entropy, equal to $1/N$, putting a lower bound on the communication required.

\section{Conclusion}

We conclude this letter by summarizing the new results. Theorem \ref{thm:most general constrained inequality} concerns pure states and establishes, for general values of $N\geq 4$, that inside certain faces of $\Sigma_N$, defined by requiring one constituent system, say $X_1$, to have vanishing mutual information with all others, there is a strictly positive lower bound on the distance from the apex to any entropy vector corresponding to a pure state with $H(X_1) \neq 0$. 

Corollary \ref{cor: constrained inequality N constarints} concerns arbitrary states for $N$-partite system with $N \geq 3$. In particular, for the case $N=3$, it entails a positive lower bound on the distance from the apex to any realizable entropy vector on any given ray within the $4$-dimensional face of $\Sigma_3$ defined by
\begin{align*} 
III_{XY}=0\ \quad \text{for all} \ X,Y,\;\;\mbox{and $H(ABC) \neq 0$}.
\end{align*}
This answers, in particular, a question posed by A. Winter \cite{noauthor_qmath_2022} concerning the possibility of down-scaling certain realizable entropy vectors. For general values of $N\geq 3$, Corollary \ref{cor: constrained inequality N constarints} provides non-homogeneous bounds \eqref{newconstraint}, which rule out down-scaled versions of realisable entropy vectors -- such as those presented in Appendix \ref{app:explicit expressions for states}. It follows that $\Sigma_N^*$ is not a cone for $N\geq 3$. On the other hand, the closure of $\Sigma_N^*$ \emph{is} a cone \cite{pippenger_inequalities_2003}, so it likewise follows that $\Sigma_N^*$ is not closed for $N\geq 3$. This confirms a previous statement from \cite{linden_new_2005} and solves an open problem from \cite{pippenger_inequalities_2003}. 

We emphasise that our results apply to the case of finite dimensional as well as infinite dimensional state spaces, provided the states in question have well-defined entropies. 

In section \ref{sec:application to physics}, we highlighted the potential impact in macroscopic systems, quantum information theory, and quantum cryptography -- pointing to the importance of a further investigation of the shape of $\Sigma_N^*$.

\acknowledgements
 We acknowledge financial support from the VILLUM FONDEN via the QMATH Centre of Excellence (Grant No.10059), the European Research Council (ERC Grant Agreement No. 81876) and the Novo Nordisk Foundation (grant NNF20OC0059939 ‘Quantum for Life’).

\bigskip

\appendix

\begin{appendix}

\section{Proofs of Lemmas \ref{lemma: lower bound V1111}, \ref{lemma: lower bound Va111} and \ref{lemma:product inequality}}\label{app:proof of lemmas}

In this section we provide proofs of Lemma 1, Lemma 2, and Lemma 3 in succession.

\begin{proof}[Proof of Lemma \ref{lemma: lower bound V1111}]

Using \eqref{eq: V has norm 1} we have
\begin{align}\label{eq: id1}
|V_{1...1}|^2 & = 1-\sum_{(x_1...x_N) \neq (1,...,1)} |V_{x_1...x_N}|^2\,.
\end{align}
Since $(x_1...x_N) \neq (1,...,1)$ implies that there exists some index $i$ such that $x_i >1$ and all terms in the above sum are non-negative, it follows that
\begin{align}
& \sum_{(x_1...x_N) \neq (1,...,1)} |V_{x_1...x_N}|^2 \notag \\ \label{eq:counting inequality}
& \leq \sum_{i = 1}^N \quad \sum_{x_1,...,x_N | x_i > 1} |V_{x_1... x_N}|^2\,.
\end{align}
 Identifying the last sum as single-party density matrix elements by use of (\ref{eq: density matrix 1 in terms of V components}), the right-hand side in the inequality above can be rewritten as 
\begin{align}\label{eq: est2}
\sum_{i=1}^N \sum_{x_i > 1} (\rho_{X_i})_{x_i x_i}= \sum_{i=1}^N \sum_{x_i>1} \lambda_{x_i}^i = 
\sum_{i=1}^N \epsilon_i = \varepsilon\,.
\end{align}
Combining \eqref{eq:counting inequality} and \eqref{eq: est2} with \eqref{eq: id1}, the claimed inequality follows.
    
\end{proof}

\begin{proof}[Proof of Lemma \ref{lemma: lower bound Va111}]

In a similar way as in the proof of Lemma \ref{lemma: lower bound V1111}, we have
\begin{align}
& \sum_{x_1>1} |V_{x_1 1 ... 1}|^2 = \sum_{x_1>1} \sum_{x_2,...,x_N} |V_{x_1 x_2  ... x_N}|^2 \notag \\
& \qquad \qquad \qquad \quad  - \sum_{x_1>1} \sum_{(x_2,...,x_N) \neq (1,...,1)} |V_{x_1 x_2  ... x_N}|^2 \notag\\
& \qquad \quad \geq \sum_{x_1>1} \sum_{x_2,...,x_N} |V_{x_1 x_2  ... x_N}|^2 \\
& \qquad \qquad -\sum_{i\geq2}^N \quad \sum_{x_1,...,x_N | x_1,x_i > 1} |V_{x_1 x_2 ... x_N}|^2\,.\label{eq: est3}
\end{align}

Using (\ref{eq: density matrix 1 in terms of V components}) and (\ref{eq: bipartite density matrix 1 2 in terms of V components}), we see that
\begin{align*}
& \sum_{x_1>1} \sum_{x_2,...,x_N} |V_{x_1 x_2  ... x_N}|^2 = \sum_{x>1} (\rho_{X_1})_{xx} = \sum_{x>1} \lambda_x^1=\epsilon_1\,, 
\end{align*}
while \eqref{eq: bipartite density matrix 1 2 in terms of V components} and \eqref{eq:components of relevant product state for main thm} yield 
\begin{align*}
& \sum_{i\geq2}^N \quad \sum_{x_1,...,x_N | x_1,x_i > 1} |V_{x_1 x_2 ... x_N}|^2 \\
& \quad = \sum_{i\geq2}^N \sum_{x_1,x_i > 1} (\rho_{X_1 X_i})_{x_1 x_i \ x_1 x_i} = \sum_{i\geq2}^N \sum_{x_1,x_i > 1} \lambda_{x_1}^1 \lambda_{x_i}^i \\
& \qquad = \sum_{i \geq 2}^N \epsilon_1 \epsilon_i=\epsilon_1(\varepsilon-\epsilon_1)\,.
\end{align*}

Inserting the last two identities into \eqref{eq: est3}, we get 
\begin{align*}
& \sum_{x_1>1} |V_{x_1 1  ... 1}|^2  \geq  \epsilon_1-\epsilon_1(\varepsilon-\epsilon_1) = \epsilon_1 (1+\epsilon_1-\varepsilon)\,,
\end{align*}
which proves \eqref{eq:lower boung Vi111} and hence the Lemma. 

\end{proof}

\begin{proof}[Proof of Lemma \ref{lemma:product inequality}]

For fixed $a>1$ we have 
\begin{align*}
& |V_{1 1...1}|^2|V_{a 1...1}|^2 = |V_{1 1...1} V_{a 1...1}^*|^2 \\
& \quad = \left|  \sum_{x_2,...,x_N } V_{1 x_2 ... x_N} V_{a x_2 ... x_N}^* \right. \\
& \qquad \qquad \qquad    \left. - \sum_{(x_2,...,x_N) \neq (1,...,1)} V_{1 x_2 ... x_N} V_{a x_2 ... x_N}^*  \right|^2 \\
& \qquad = \left| \sum_{(x_2,...,x_N) \neq (1,...,1)} V_{1 x_2 ... x_N} V_{a x_2 ... x_N}^*  \right|^2\,,
\end{align*}
where in the last step it has been used that (\ref{eq: density matrix 1 in terms of V components}) and (\ref{eq:diagonalized density matrix}) imply that $\sum_{x_2,...,x_N } V_{1 x_2 ... x_N} V_{a x_2 ... x_N}^* = (\rho_{X_1})_{1 a}=\lambda_1^1 \delta_{1 a}=0$. 

We now use the Cauchy-Schwarz inequality to bound the last expression from above by  
\begin{equation*}
\sum_{(x_2,...,x_N) \neq (1,...,1)} |V_{1x_2...x_N}|^2 \sum_{(x_2,...,x_N) \neq (1,...,1)} |V_{a x_2...x_N}|^2\,.
\end{equation*}
Using again relations (\ref{eq: density matrix 1 in terms of V components}) and (\ref{eq:diagonalized density matrix}) to rewrite these sums in terms of matrix elements of $\rho_{X_1}$, we obtain 
\begin{align*}
&\; \quad |V_{1 1...1}|^2|V_{a 1...1}|^2 \\
& \leq \left( (\rho_{X_1})_{1 1}-|V_{1 1 ... 1}|^2 \right) \left( (\rho_{X_1})_{a a}-|V_{a 1 ... 1}|^2 \right) \\
& = \left( 1-\epsilon_1-|V_{1 1 ... 1}|^2 \right) \left( \lambda_a^1-|V_{a 1 ... 1}|^2 \right)\,.
\end{align*}
Rewriting this inequality as 
\begin{align*}
(1-\epsilon_1) |V_{a 1 ... 1}|^2 \leq (1-\epsilon_1-|V_{1 ... 1}|^2) \lambda_a^1
\end{align*}
and summing over $a>1$ then yields  
\begin{equation*}
(1-\epsilon_1) \sum_{x_1>1} |V_{x_1 1 ... 1}|^2 \leq (1-\epsilon_1-|V_{1 ... 1}|^2)\epsilon_1\,.
\end{equation*}
Finally, applying the lower bound from Lemma \ref{lemma: lower bound V1111} on the right-hand side gives the desired inequality \eqref{eq:product inequality eq}, which proves the lemma. 
\end{proof}

\section{Further elaboration and refinement of the main theorem} \label{app:lower bound for max entropy}

In this section we discuss a slight refinement of Theorem \ref{thm:most general constrained inequality}, or rather Corollary \ref{cor: constrained inequality N constarints}, concerning the entropy-vector $\vec H=\vec H_\rho$ for a system $\mathcal{N}=\{X_1,...,X_N\}$ in a state $\rho$ that is not necessarily pure.

\vspace{0.5 cm}
Assume $\vec H$ fulfills either of the following conditions:
\begin{itemize}
    \item[i)] There exists an $i \in \{1,...,N\}$ such that $H(X_i) \neq 0$ and $I(X_i:X_j)=0$ for $j\neq i$ and $H(X_i)+H(\mathcal{N})= H(\mathcal{N} \backslash X_i)$\,,
\item[ii)] $H(\mathcal{N}) > 0$ and $H(X_i)+H(\mathcal{N})=H(\mathcal{N} \backslash X_i)$ for all $i \in \{1,...,N\}$\,. 
\end{itemize}

We then claim the bound 
\begin{align} \label{eq:refinement of main result largest entropy}
\max( \{H(X_1),...,H(X_N),H(\mathcal{N})\} ) & \ > \ h \left( \frac{1}{2 N^{'}} \right)
\end{align}
 holds, where $N^{'}$ is the number of non-zero elements among $H(X_1),...,H(X_N),H(\mathcal{N})$. 

Indeed, we first note that the constituent systems with vanishing entropy can be disregarded since they are in a pure state in product with the rest and do not contribute to any of the component entropies of $\vec H$. By purification of the resulting system we obtain an $N'$-partite system $\mathcal N'$ in a pure state that satisfies the conditions of Theorem \ref{thm:most general constrained inequality} by using \eqref{eq:entropy bound larg eigenvalue} and either of the assumptions i) or ii). In case of i) $X_1$ is replaced by $X_i$ while in case ii) $X_1$ is replaced by the purifying system $E$ whose entropy is $H(\mathcal N)$. From the proof of Theorem \ref{thm:most general constrained inequality} as applied to the system $\mathcal N'$ we find that $\varepsilon\equiv \sum_{i}\epsilon_i + \epsilon_E>\frac 12$, with obvious notation. Using the intermediate inequality from the elementary entropy bound (\ref{eq:entropy bound larg eigenvalue}) and the fact that  $\max \left[ h(x), \ - \log(1-x) \right]$ is an increasing function of $x \in [0,1]$, this implies that the left-hand side of \eqref{eq:refinement of main result largest entropy} is strictly lower bounded by 
$$
\max \left[ h \left(\frac{\epsilon}{N'} \right), \ - \log \left(1-\frac{\epsilon}{N'} \right) \right] \geq h \left( \frac{1}{2N'} \right)
$$
as claimed. 

This result entails, for example, that vectors on the open line-segment of $\ell$ between the origin and $\alpha (1,1,1,2,2,2,1)$, where $\alpha =h \left( \frac{1}{8} \right) \approx 0.54$, cannot be realized as entropy vectors of a quantum system. Note also that condition ii) above coincides with the assumptions of  Corollary \ref{cor: constrained inequality N constarints}. 

\bigskip
Next, we discuss further restrictions on states satisfying the conditions in Theorem \ref{thm:most general constrained inequality}. Suppose $N=4$ and let $\rho$ be a pure state which satisfies
$$
H(X_1) >0 \;\; \mbox{and}\;\;I(X_1:X_i)=0 \; \mbox{for}\;i \in \{2,3,4\}.
$$
Assuming w.l.o.g. that $H(X_2) \leq H(X_3) \leq H(X_4)$, we then have that 
\begin{align} \label{eq:app B special entropy less than others}
H(X_1) \leq H(X_i) \; \ \mbox{for} \ \; i \in \{2,3,4\}
\end{align}
as well as  
\begin{align*}
H(X_4)  \geq h \left( \frac{1}{8} \right) \quad \mbox{and}\quad 
H(X_3)  \geq \frac{1}{2} h \left( \frac{1}{8} \right)\,. 
\end{align*}

To establish these inequalities we use that any reduced state of a total pure state have equal entropy to its complement, sub-additivity, and the assumption $I(X_1 : X_4)=0$, together with the fact that $\rho$ is pure, imply that
\begin{align} \label{eq:app B ineq for prove X1 less X2}
H(X_1)+H(X_4) \leq H(X_2)+H(X_3)\,,
\end{align}
and hence $H(X_1) \leq H(X_2)$ which proves (\ref{eq:app B special entropy less than others}). The inequality (\ref{eq:app B ineq for prove X1 less X2}) also entails
\begin{align*}
H(X_3) \geq \frac{1}{2} H(X_4) \geq \frac{1}{2} h \left( \frac{1}{8} \right)
\end{align*}
by (\ref{eq:refinement of main result largest entropy}), thus proving the other two inequalities. 

By a similar argument we get for general $N \geq 4$ that, if the conditions of Theorem \ref{thm:most general constrained inequality} are satisfied and we assume $H(X_2)\leq\dots \leq H(X_N)$, then
$$H(X_{N-1}) \geq \frac{1}{N-2} H(X_N) \geq \frac{1}{N-2} h \left( \frac{1}{2N} \right)\,.$$

\section{Constructing states satisfying the conditions from the main theorem} \label{app:explicit expressions for states}
In this section, we provide ways of constructing families of states, whose entropy vectors satisfy the conditions of Theorem \ref{thm:most general constrained inequality}. 

\bigskip
We first note that analogues of Pippenger's state $\eta$ mentioned in the introduction and fulfilling \eqref{eq: line3} for $N=4$ can be constructed for arbitrary $N$. An example that satisfies \eqref{eq: line3} for all pairs $\{X_i,X_j\}, 1\leq i<j\leq N$, is  
obtained by choosing the individual state spaces $\mathcal H_{X_i}$ to be identical of dimension $N$ with an orthonormal basis $\ket{k},\, k=0,1,\dots, N-1$, and defining 
\begin{align} 
\ket{V_N} & = 
\quad \frac{1}{N} \sum_{k=0}^{N-1} \ket{k,...,k} \notag \\  \label{eq:explicit N party special state}
& + \frac{1}{N} \sqrt{\frac{2}{(N-2)!}} \sum_{\sigma \in A_N} \ket{\sigma(0),...,\sigma(N-1)}\,,
\end{align}
where $\sigma$ runs over the set  $A_N$ of even permutations of $\{0,1,...,N-1\}$. For the state $\mu = \ket{V_N} \bra{V_N}$ one finds that
\begin{equation}\label{eq: etamut}
\mu_{X_i} = \sum_{k=0}^{N-1} \frac{1}{N} \ket{k} \bra{k} \quad \mbox{and}\quad 
\mu_{X_iX_j} = \mu_{X_i} \otimes\mu_{X_j} \,.
\end{equation}
It follows that $H(X_i)=\log N$ for all $i$ and from additivity of entropy for product states that \eqref{eq: line3} is fulfilled for all pairs $\{X_i,X_j\}$.  

Another, less complex, construction applies to the case where \eqref{eq: line3} only is required to hold for $i=1$ (relevant for Theorem \ref{thm:most general constrained inequality}) and $j\neq 1$. In this case, we assume the spaces $\mathcal H_{X_i}$ 
have dimension $N-1$ with orthonormal basis $\ket{k}, k=0,1,\dots N-2$, and we set 
\begin{align}
&\ket{W_N} = 
\, \frac{1}{N-1} \sum_{k=0}^{N-2} \ket{0,k,...,k} 
 + \frac{1}{N-1} \sum_{k=1}^{N-2}\sum_{l=0}^{N-2}\notag\\ 
 & \ \ket{k,l,l+k,l+k+1,\dots,l+N-2, l+1,\dots,l+k-1} , \notag
\end{align}
where addition inside kets is mod $N-1$. By inspection one finds that \eqref{eq: etamut} likewise holds for $\omega=\ket{W_N} \bra{W_N}$ when $i=1$ (with $N$ replaced by $N-1$) and so the corresponding mutual informations vanish as claimed. Incidentally, for $N=4$ all six mutual informations vanish for this state and its entropy vector is hence proportional to $\vec H(\ket{V_4}\bra{V_4})$ as well as to $\vec H(\eta)$. For the reader's convenience, we record the  explicit expressions for $\ket{V_4}$ and $\ket{W_4}$:
 
\begin{align*}
\ket{V_4} & = \frac{1}{4} \Big(\ket{0000}+\ket{0123}+\ket{0231}+\ket{0312} +\ket{1111}\\ +&\ket{1032}+\ket{1320}+\ket{1203}+\ket{2222}+\ket{2301}+\ket{2013}\\ &+\ket{2130} +\ket{3333}+\ket{3210}+\ket{3102}+\ket{3021} \Big)\,, 
\end{align*}
\begin{align*}
\ket{W_4} & = \, \frac{1}{3} \Big(\ket{0000}+\ket{0111}+\ket{0222}+\ket{1012} \\ 
& +\ket{1120}+\ket{1201}+\ket{2021}+\ket{2102}+\ket{2210} \Big)\,.
\end{align*}

As noted at the end of Section \ref{sec:presenting entropy inequality}, the set of entropy vectors in $\Sigma_N^*$ satisfying the conditions of Corollary \ref{cor: constrained inequality N constarints} has dimension at most $2^N-N-1$, which equals $4$ for $N=3$. We now construct by simple means such a $4$-parameter family for arbitrary $N\geq 3$. We concentrate on the case $N=3$, since the extension to $N>3$ is straightforward.  

As a first step we generalize the state $\ket{V_4}$ given above to 
\begin{equation} \label{eq:explicit N party special state variable coeff}
\ket{\tilde V_{4}}  = 
 \frac{1}{2} \sum_{i=0}^{3} a_i \ket{iiii} + \frac{1}{2} \sum_{\sigma \in A_4} a_{\sigma(3)} \ket{\sigma(0)\sigma(1)\sigma(2)\sigma(3)}\;,
\end{equation}
where the numbers $\{a_i\}_{i=0}^{3}$ satisfy $\sum_{i=0}^{3} |a_i|^2=1$.
The pure state $\ket{\tilde V_{4}} \bra{\tilde V_{4}}$ then still fulfills $I(X_4:X_j)=0$ for $j=1,2,3$ while  
\begin{equation*}
 H(X_4)= \alpha \;\;\mbox{and\; $H(X_j)=2$ for $j=1,2,3$}\,,
\end{equation*}
where $\alpha =  - \sum_{i=0}^{3} |a_i|^2 \log (|a_i|^2)$ can assume any value in the interval $[0,2]$. For the corresponding $3$-partite state $\rho_{ABC}=Tr_D [\ket{\tilde V_{4}} \bra{\tilde V_{4}}$, where we use the notation $A,B,C,D$ instead of $X_1,X_2,X_3,X_4$ , we thus obtain 
the following $1$-parameter family of entropy vectors by using that any reduced state of a total pure state have equal entropy to its complement
\begin{align*}
\vec{H}_{\tilde\rho}=(2,2,2,2+\alpha,2+\alpha,2+\alpha,\alpha)\,.
\end{align*}

Next, let $B'$ and $C'$ be two new systems with state space $\mathcal H_{B'} = \mathcal H_{C'}$ having orthonormal basis states $\ket{l},\; l=1,\dots,M$, and let the pure state $\rho'_{B'C'}=\ket{V'}\bra{V'}$ on the composite system $B'C'$ be given by 
$$\ket{V'} =  \sum_k a'_l\ket{ll}\,,$$
where $\sum_l |a'_l|^2 = 1$. 
Then 
$$H_{\rho'_{B'C'}}(B')= H_{\rho'_{B'C'}}(C') = - \sum_l |a'_l|^2\log |a'_l|^2\,.$$

Considering the product state $\tau = \tilde\rho_{ABC}\otimes \rho'_{B'C'}$ as a $3$-partite state with constituent systems $A, BB', CC'$, it now follows from additivity of entropy that the entropy vector of this product state is 
$$\vec H_\tau = \vec H_{\tilde\rho} + (0,\beta,\beta, 0, \beta,\beta,0)\,,$$
which is a $2$-parameter family of states with $\beta := - \sum_l |a'_l|^2\log |a'_l|^2$ varying in the interval $[0,\log M]$ and fulfilling the conditions of Corollary \ref{cor: constrained inequality N constarints}. Repeating this last construction with new systems $A''C''$ and $A'''B'''$, we finally obtain a $4$-parameter family in $\Sigma_3^*$,
\begin{align*}\vec H(\alpha,&\beta,\gamma,\delta) = (2+\gamma+\delta, 2+\beta+\delta,2+\beta+\gamma,\\
& 2+\alpha+\gamma+\delta,2+\alpha+\beta+\delta,
2+\alpha+\beta+\gamma,\alpha)\,,
\end{align*}
with the claimed properties. For $N\geq 4$ the method can be further refined and obviously leads to a larger number of parameters than $4$, but we refrain from elaborating on this issue here.

\end{appendix}




\providecommand{\noopsort}[1]{}\providecommand{\singleletter}[1]{#1}%

\end{document}